\documentclass[10pt,aps,prl,twocolumn,superscriptaddress,floatfix,nofootinbib]{revtex4-2}

\usepackage{amsmath, amsfonts, amssymb }
\usepackage{amsthm}

\usepackage{bm,bbm, braket, accents}
\usepackage{graphicx}   % need for figures
\usepackage{soul}
\usepackage[usenames,dvipsnames]{xcolor}
\usepackage{lipsum} 
\usepackage{tikz}
\usetikzlibrary{quantikz}
\usepackage{pgfplots}
\pgfplotsset{compat=newest}
\usepackage{comment}

\usepackage[normalem]{ulem}

\usepackage[breaklinks=true]{hyperref}
\hypersetup{
  colorlinks   = true, %Colour links instead of boxes
  urlcolor     = blue, %Colour for external hyperlinks
  linkcolor    = blue, %Colour of internal links
  citecolor    = red %Colour of citations
}

\usepackage[capitalise]{cleveref} % must be loaded after hyperref
\crefformat{equation}{Eq.~(#2#1#3)} % These change 'equation' to Eq., more PRL or PRA-style
\crefformat{section}{Sec.~#2#1#3} % These change 'equation' to Eq., morePRL or PRA-style
\Crefformat{equation}{Equation~(#2#1#3)}
\crefformat{figure}{Fig.~#2#1#3}
\crefrangeformat{equation}{Eqs.~#3(#1)#4--#5(#2)#6}
\Crefformat{section}{Section~#2#1#3}

\newtheorem{theorem}{Theorem}

\usepackage{orcidlink}

\usepackage{floatrow}

\makeatletter
\def\@bibdataout@aps{%
 \immediate\write\@bibdataout{%
  @CONTROL{%
   apsrev41Control,author="08",editor="1",pages="0",title="0",year="1",eprint="1"%
  }%
 }%
 \if@filesw
  \immediate\write\@auxout{\string\citation{apsrev41Control}}%
 \fi
}%
\makeatother % Phew.
\newcommand{\dg}{^\dagger}
\newcommand{\smallfrac}[2]{\mbox{$\frac{#1}{#2}$}}
\newcommand{\half}{\smallfrac{1}{2}}

\newcommand{\op}[2]{\ket{#1}\!\bra{#2}}
\newcommand{\expt}[1]{\langle{#1}\rangle}\newcommand{\var}{\text{var}}
\newcommand{\nmax}{D}

\newcommand{\E}[1]{\hat{\mathbb{E}}_{#1}}
\newcommand{\Ep}[2]{{\hat{\mathbb{E}}}_{#1}\!\left( #2 \right)}

\begin{document}
%============================================

\title{ 
Noise constraints on sensitivity scaling in super-Heisenberg quantum metrology
}

\author{Noah Lordi\orcidlink{https://orcid.org/0009-0005-6970-2240}}
\thanks{Corresponding author: John-Wilson-6@colorado.edu}
\affiliation{Department of Physics, University of Colorado, Boulder, Colorado 80309, USA}
\affiliation{These authors contributed equally to this work. }

\author{John Drew Wilson\orcidlink{0000-0001-6334-2460}}
\thanks{Corresponding author: John-Wilson-6@colorado.edu}
\affiliation{Department of Physics, University of Colorado, Boulder, Colorado 80309, USA}
\affiliation{These authors contributed equally to this work. }
\affiliation{JILA, University of Colorado, Boulder, Colorado 80309, USA}

\author{Murray J. Holland\orcidlink{0000-0002-3778-1352}}
\affiliation{Department of Physics, University of Colorado, Boulder, Colorado 80309, USA}
\affiliation{JILA, University of Colorado, Boulder, Colorado 80309, USA}

\author{Joshua Combes\orcidlink{0000-0003-3199-5744}}
\affiliation{Department of Electrical, Computer and Energy Engineering, University of Colorado, Boulder, Colorado 80309, USA}
\affiliation{School of Physics and School of Mathematics, The University of Melbourne, VIC 3010, Australia}

\date{\today}

%============================================
\begin{abstract}
Quantum-enhanced metrology surpasses classical metrology by improving estimation precision scaling with a resource $N$ (e.g., particle number or energy) from $1/\sqrt{N}$ to $1/N$.
Through the use of nonlinear effects, Roy and Braunstein~\cite{Roy2008} derived a $1/2^N$ scaling.
However, later works argued this exponential improvement is unphysical and that even modest gains, like $1/N^2$, may vanish under noise. 
We show that, in the presence of small errors, the nonlinear interactions enabling metrological enhancement induce emergent errors. 
The errors propagate through the sensing protocol and are magnified proportional to any intended non-linear enhancement.  
We identify a critical value of the parameter to be estimated, for a fixed error, below which the emergent errors can be avoided.
\end{abstract}
%============================================

\maketitle

The idea of using exotic quantum states, such as squeezed states, to enhance precision measurements is several decades old ~\cite{Caves1981,Bondurant1984,Yurke1986a,Yurke1986b,Kitagawa1991, Wineland1992, kitagawa1993,Holland1993}. These ideas have been experimentally implemented in noisy real-world environments, leading to improvements in gravitational wave detection~\cite{aasi2013,tse2019} and timekeeping with atomic clocks~\cite{vuletic2020}. It is expected that quantum metrology will continue to prove useful in a variety of applications~\cite{Degen_RMP_2017}.

A central aim of quantum metrology is to estimate a parameter $\varphi$ with precision $\Delta \varphi$, as summarized in \cref{fig:summary}. 
A probe state, $\ket{\psi}$, acquires information about $\varphi$ via unitary evolution, $\hat{U}(\varphi) = \exp[-i \varphi \hat{G}]$, where $\hat{G}$ is a Hermitian operator.
Then, measurements on the final state $\ket{\Psi} = \hat{U}(\varphi) \ket{\psi}$ yield an estimate of $\varphi$.
The achievable precision $\Delta \varphi$ depends on the resource $N = \expt{\psi|\hat G|\psi}$.
For the photon number operator, $\hat{G} = \hat{n}$, classical states achieve $\Delta \varphi \sim 1/\sqrt{N}$~\cite{Caves_SQL_1980}. 
Quantum resources improve this to the Heisenberg limit (HL), $\Delta \varphi = 1/N$~\cite{Holland1993}. 
Nonlinear metrology aims to surpass the HL by using a nonlinear $\hat{G}$, e.g., $\hat{G} = \hat{n}^2$~\cite{Luis2004,BeltranLuis2005,Roy2008,Boixo2007,Boixo2008,RivasLuis2010,Napolitano2011,HallWiseman2012,Braun2018,Johnsson_2021}, thereby acheiving a so-called super-Heisenberg limited scaling.
Notably, \citet{Roy2008} showed that an $N$-qubit Hamiltonian can achieve $\Delta \varphi = 1/2^N$, enabling exponential precision gains.

\begin{figure}[!ht]
    \centering
    \includegraphics[width=\columnwidth]{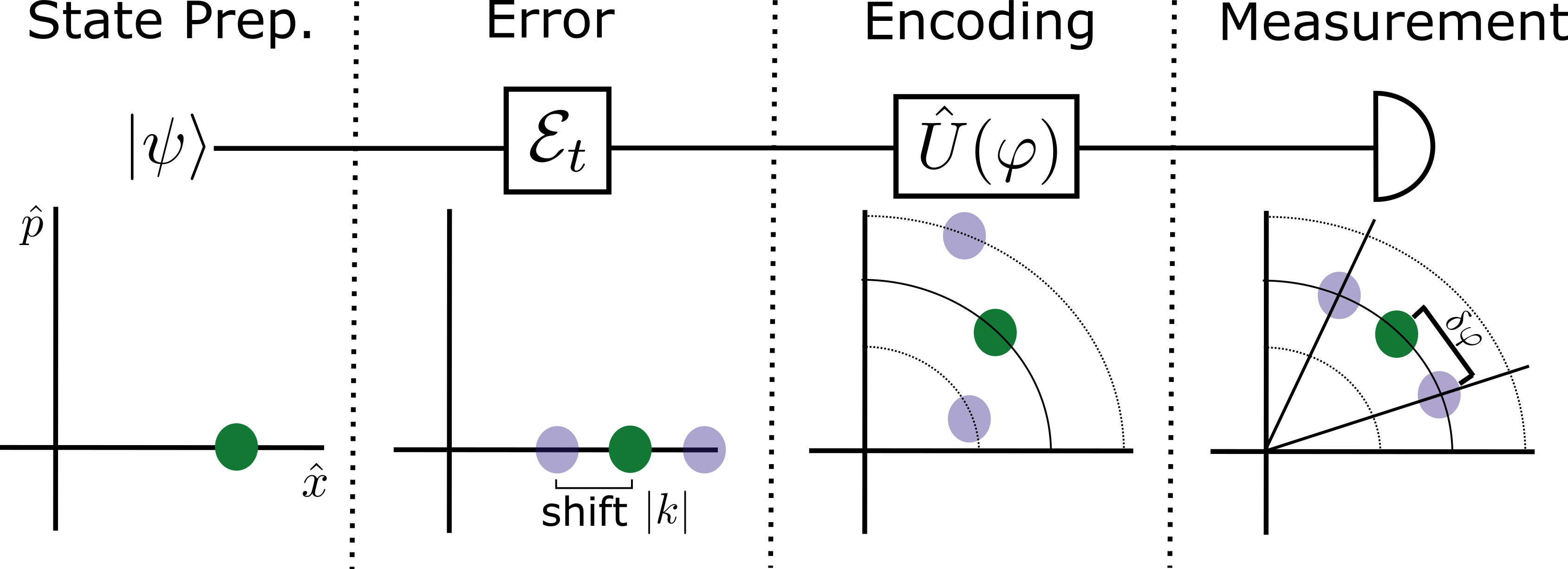}
    \caption{ A cartoon of noisy nonlinear metrology in the harmonic oscillator phase space. 
    State Prep.: A known probe state, $\ket{\psi}$, is prepared.
    Error: All errors that occur during state preparation, encoding, and measurement can be combined into a single channel $\mathcal{E}_t$. 
    In the graphic the state $\ket{\psi}$ is subject to an error which shifts the mean number of the state by $\pm k$. 
    Encoding: The post-error state evolves through a unitary $U(\varphi)$ generated by the Hamiltonian, $g(\hat{n})$, where $g$ is a nonlinear function of $\hat{n}$, e.g. $\exp(\hat n)$. This causes the parameter $\varphi$ to be encoded into the quantum state.
    Measurement: The error causes the quantum state to suffer an additional ``phase error'' upon measurement and readout. 
    In the cartoon, we depict the case that initial shift error is corrected, yet this ``phase error'' persists.
    We show the limits in which this error can be avoided, and the limits in which it dominates.
    }
    \label{fig:summary}
\end{figure}

This demonstrates the allure of nonlinear metrology, but many are reasonably skeptical that this enhanced sensitivity could ever be achieved in a realistic scenario with noise. \citet{Boixo2007} argued that the exponential scaling is un-physical based on the requirement to couple the $N$ constituent qubits to a rank-$N$ tensor field.
In that same work,~\citet{Boixo2007,Boixo2008} suggests that it might be possible to use a $z$-body interaction, e.g. $\hat G = \hat{n}^z$, to achieve an $\Delta \varphi \propto 1/N^z$ asymptotic enhancement in sensitivity. Subsequent work examined the degradation of enhanced sensitivity for specific noise channels see e.g. Refs.~\cite{ReyJiangLukin2007,ZwierzWiseman2014,LuisRivas2015,RossiParis2016,Beau2017,Rams2018,Czajkowski_2019,ImaiSmerzi2024, RiberiViola2024}. 
As the opportunity for nonlinear enhancement grows closer to experimental reality~\cite{Napolitano2011,Sewell2014,Xinfang2018,Huang2025}, there are unfortuantely no general limits of nonlinear metrology in the presence of noise.

In this Letter, we develop a framework for analyzing noisy nonlinear metrology that is applicable to  many physical systems and most noise. We establish sufficient conditions for nonlinear metrology to be useful in the presence of quantum noise and, noteably, when it might cause more harm than good. Specifically, we show that while nonlinear interactions enable enhancement, they also induce emergent errors—even when original errors are corrected. From these emergent errors, we show linear metrology only suffers a global phase--which does not harm sensing capabilities, and provide a sufficient condition for nonlinear metrology to avoid these errors.
When considering randomly chosen probe states, our findings reveal a phase transition, between usable and unusable regimes, governed by the interplay between nonlinearity in the Hamiltonian and initial error magnitude.

To make formal the earlier discussion of linear versus nonlinear metrology, we begin by considering a Hilbert space, $\mathcal{H}$, with orthonormal basis states identified by the set $\mathbb{S}$.
On this space, we define a generalized number operator 
\begin{equation}\label{eq:number_op}
\hat{n} = \sum_{n\in \mathbb{S}} n \op{n}{n},
\end{equation}
where $n$ is taken from the set $\mathbb{S}$, some examples of $\mathbb{S}$ are given in \cref{tab:spaces}.
We use this number operator to construct a Hamiltonian that is a function of $\hat{n}$
\begin{equation}\label{eq:generator}
\hat{G} = g(\hat{n}) \equiv \sum_{n\in \mathbb{S}} g(n) \op{n}{n}\, ,
\end{equation}
where $g$ is smooth, differentiable, and real valued.
If $g$ is a linear function then \cref{eq:generator} corresponds to linear metrology, while all other $g$'s can be regarded as nonlinear metrology.    
The Hamiltonian in \cref{eq:generator} will impart a phase, $\varphi$, onto a probe state, $\ket{\psi}$, via the unitary
\begin{equation}\label{eq:Unitary}
\hat{U}_g(\varphi) \equiv \exp( -i \varphi \hat{G} ) = \sum_{n\in \mathbb S} \mathrm{e}^{ -i \varphi g(n) } \op{n}{n}.
\end{equation}
The ultimate limit in sensing precision with a pure state is given by the quantum Cram\`er Rao bound:
\begin{equation}\label{eq:QCRB}
    \Delta \varphi \ge \frac{1}{ \sqrt{ \nu \ I\big(\varphi|\hat G,\psi\big)}} 
\end{equation}
where $\nu$ is the number of independent trials, and $I(\varphi|\hat{G},\psi) = 4 \expt{\psi|\hat{G}^2|\psi} - 4 \expt{\psi|\hat{G}|\psi}^2$ is the quantum Fisher information (QFI)~\cite{BraunCaves1994,Pezze2014,liu2019} for a pure state.
The maximum  QFI   $I(\varphi|G,\psi_{\rm opt}) = [g(n_{\rm max})-g(n_{\rm min})]^2$ is obtained for the probe state $\ket{\psi_{\rm opt}} \equiv (\ket{n_{\rm min}}+\ket{n_{\rm max}})/\sqrt{2}$ with $n_{\rm max}> n_{\rm min}$ both in $\mathbb S$~\cite{DegenCondition}.

Our aim now is to show when errors in nonlinear metrology are harmful and when they are harmless to the overall metrological usefulness.
To achieve this goal, we introduce a non-Hermitian operator basis to decompose any error channel or Lindblad evolution.
This will allow us to demonstrate which components of the errors have a non-trivial interplay with the nonlinear encoding step, thereby leading to emergent phase errors which may not be simply corrected~\cite{zhou2018}.

Errors accumulate continuously in time and can always be modeled by quantum channels, where we decompose the noisy encoding operation at some fixed time $t$ as $\tilde{\mathcal U}_t = \mathcal U_t \circ \mathcal E_t$ where $\mathcal{U}_t$ is the ideal encoding and $\circ$ denotes channel composition. 
As explained in \cref{fig:summary}, we can commute all errors from the state preparation, encoding, and measurement to just before the ideal encoding step. 
That yields $\mathcal{U}_t\circ \mathcal{E}_t = \mathcal{E}_m \circ \mathcal{U}_t\circ\mathcal{E}_{e} \circ \mathcal{E}_p$, where $\mathcal{E}_p$, $\mathcal{E}_{e}$, and $\mathcal{E}_m$ represent errors during preparation, encoding, and measurement, respectively.
 
Next, we decompose the Kraus operators of this error channel into an operator basis. We call this basis the {\em error basis}, and it is given by
\begin{equation}\label{eq:error_basis}
 \Ep{k}{\Phi} \equiv \left\{ \begin{array}{ll}
\mathrm{e}^{i\Phi \hat{n}} \Sigma_{|k|}^- & \text{ for } k < 0 \\
\Sigma_{|k|}^+ \mathrm{e}^{-i\Phi \hat{n}} & \text{ for } k \ge 0
  \end{array}\right. ,
\end{equation}
where the shift down and up operator are
\begin{equation}\label{eq:Errors}
\hat{\Sigma}^-_1 \equiv \sum_{n\in\mathbb{S}} \op{n}{n+1}, \quad
\hat{\Sigma}_1^+ =  \big( \hat{\Sigma}_1^- \big)^\dagger , \quad \hat{\Sigma}_{|k|}^\pm =  \big (\hat{\Sigma}_1^\pm\big)^{|k|} .
\end{equation}
For brevity, we write $\E{k}$ whenever $\Phi=0$. 
The edge cases such as $n\in\mathbb{S}$ but $n+1\neq\mathbb{S}$ are shown in Supplemental Material (SM)~\cite{suppMat}, and we note this operator basis is, in general, over-complete i.e. $\int d\Phi \sum_{k} \E{k}^\dagger(\Phi)\Ep{k}{\Phi} \propto \hat{\mathbb{I}}$.

We'd like to see how the nonlinearity of $g$ amplifies the effects of an error.
Consider an arbitrary element from our error basis, $ \Ep{k}{\Phi}$, and ``pull'' it through the encoding unitary:
\begin{equation}\label{eq:circuit}
\resizebox{0.88\columnwidth}{!}{
\raisebox{0.25em}{
\begin{quantikz}[row sep=0.295cm,column sep=0.2cm]
    & \gate[][0.7cm][0.7cm]{\Ep{k}{\Phi}} & \gate[][0.7cm][0.7cm]{\hat{U}_g(\varphi)}& \qw
\end{quantikz}
\hspace{-0.1 em} \raisebox{-0.25em}{=}
\begin{quantikz}[row sep=0.295cm,column sep=0.2cm]
    & \gate[][0.7cm][0.7cm]{\hat{U}_g(\varphi) } & \gate{\Ep{k}{\Phi}} & \gate{\hat{V}_k(\varphi)} & \qw
\end{quantikz}
},\tag{7a}
}
\end{equation}
which results in an emergent error $\hat{V}_k(\varphi)$. 
Mathematically, the emergence of  $\hat{V}_k(\varphi)$ in \cref{eq:circuit} is induced by the commutation of $\E{k}$ and $\hat{U}_g(\varphi)$,
\begin{equation}\label{eq:commute_err}
\hat{U}_g(\varphi)\Ep{k}{\Phi} = \hat{V}_k(\varphi)\ \Ep{k}{\Phi} \hat{U}_g(\varphi)\,,\tag{7b}
\end{equation}
where the form of this error is our first key result,
\begin{equation}\label{eqn:Vk}
\hat{V}_k(\varphi)\equiv \exp(-i \varphi \left[ g(\hat{n}) - g(\hat{n} - k) \right] ).\tag{7c}
\setcounter{equation}{7}
\end{equation}
The emergent phase error, $\hat{V}_k(\varphi)$, depends on the unknown phase $\varphi$ and on the nonlinearity of the function $g$ and is calculated explicitly in the SM~\cite{suppMat} for all the Hilbert spaces \cref{tab:spaces} -- some special cases were previously derived by ~\citet{Marinoff2024}.
Clearly the emergent error has no $\Phi$ dependence, so we fix $\Phi = 0$ without loss of generality. Even if the shift error, $\E{k}$, is undone, this emergent error will persist.

\begin{table}[t]
\floatbox[{\capbeside\thisfloatsetup{capbesideposition={right,top},capbesidewidth=3cm}}]{table}[\FBwidth]
{
\caption{
Hilbert spaces with basis states identified by the set $\mathbb{S}$.
We take $0\in\mathbb{N}$.
Edge cases with each $\mathbb{S}$ are dealt with in the SM~\cite{suppMat}.
}\label{tab:spaces}
}
{
\begin{tabular}{l|ll}
Hilbert Space, $\mathcal{H}$       & $\mathbb{S}$        \\ \hline
Bosonic             & $\mathbb{N}$        \\
Rigid Rotor    & $\mathbb{Z}$        \\
Spin Ensemble       & $\{0,\dots,N\}$  \\
Continuous Variable & $\mathbb{R}$                       
\end{tabular}
}
\end{table}

Let's consider some consequences of \cref{eqn:Vk}. First, for $g(\hat{n}) \propto \hat{n}$ we recover that $\hat{V}_k(\varphi) = \exp(-i \varphi k )$ is a global phase, and therefore linear metrology is unaffected. 
Second, we observe that $g(n) - g(n - k) \approx k g'(n) $ is approximately the derivative with respect to $n$ for small $k$. 
As a result, for $g(\hat{n}) \propto \hat{n}^z$ with real number $z\ge 1$ we find $\hat{V}_k(\varphi) \approx \exp(-i \varphi k z \hat{n}^{z-1} )$ is a nonlinear phase error due to the nonlinearity of $g$. 
This emergent error commutes with $\hat{U}(\varphi)$, making it hard to distinguish from the encoding. 
Furthermore, this error is hard to correct due to the dependence on both $\varphi$ and $k$, although the $k$ dependence can be heralded in principle and incorporated into the data-processing.

Now, we can assess the harm of the emergent error. 
When the parameter $\varphi$ is small, the effects of $\hat{V}_k(\varphi)$ are perturbatively small, and do not meaningfully impact nonlinear metrology.
To determine when these nonlinear errors meaningfully impact metrology, we compare when the shift error occurs before the unitary, i.e.
$\ket{\Upsilon}= \hat{U}_g(\varphi) \E{k}\ket{\psi} = \hat{V}_k(\varphi) \E{k} \hat{U}_g(\varphi) \ket{\psi}$, and when it occurs after $ \ket{\Psi} = \E{k} \hat{U}_g(\varphi) \ket{\psi}$.
Clearly, if the error happens after the encoding there is no nonlinear error. 

We argue that if the two states are similar, as diagnosed by the state fidelity, the errors do not significantly degrade the metrological utility, i.e. 
\begin{equation} \label{eq:overlap}
|\langle \Psi | {\Upsilon} \rangle|^2 = |\bra{\Psi}\hat{V}_k(\varphi) \ket{\Psi}|^2 = 1 - O( \epsilon)
\end{equation}
for $\epsilon \ll 1$. There is an important assumption that the shift error does not move amplitude out of the Hilbert space, which loosely means that the state $\ket{\psi}$ cannot have significant amplitudes below $k$. Mathematically this is a condition on $\ket{\psi}$ such that $|\langle \Psi | \Psi \rangle|^2 = 1$. \cref{eq:overlap} guarantees that, under any measurements and data processing, the resultant probability distributions will be at most of order $\epsilon$ different~\cite{normNote} and therefore the estimated values of $\varphi$ will be effectively unchanged as well.

For an arbitrary probe state, $\ket{\psi}$,~\cref{eq:overlap} can always be satisfied so long as $\varphi$ is bounded in magnitude by a sufficiently small value, so that $\hat{V}_k(\varphi)$ only perturbatively changes the state.
The value of $\varphi$ below which~\cref{eq:overlap} holds, is the {\em critical phase}.
We now derive the conditions for a critical phase. These conditions exist for all Hilbert spaces given in~\cref{tab:spaces} as discussed in the SM~\cite{suppMat}. 
Below we restrict ourselves to the bosonic case, $\mathbb{S} = \mathbb{N}$, and introduce a  ``photon" number cutoff, $\nmax$.

\begin{theorem}[critical phase for fixed error]\label{thrm:noise}
For a fixed shift error $k$, if the encoded phase $\varphi$ is sufficiently small such that 
\begin{equation}\label{eq:go_no_go}
|\varphi| \leq \frac{\sqrt{\epsilon}}{  \max_{n\leq \nmax} | g(n+k) - g(n) | },
\end{equation}
then $|\langle \Psi | {\Upsilon} \rangle|^2 = 1 - O( \epsilon)$, where $\epsilon\ll1$ is the error bound.
\end{theorem}

\begin{proof}
First, let $\Delta_k(n) \equiv g(n+k) - g(n)$. 
Now, take $k$ to be fixed, and assume $|\varphi| \leq \sqrt{\epsilon} / ( \max_{n\leq \nmax} |\Delta_k(n) | ) $ for $\epsilon \ll 1$. Notice that
\begin{equation}\label{eq:EVE}
\E{k}^{\dagger} \hat{V}_k(\varphi) \E{k} = \sum_{n=m}^\nmax \mathrm{e}^{-i \Delta_k(n) \varphi } \op{n}{n},
\end{equation}
for $m = \max(-k,0)$ because if $k<0$ then any $\ket{n}$ with $n<|k|$ is shifted out of the Hilbert space.
Now, expand the probe state in the number basis $\ket{\psi} = \sum_{n=0}^\nmax c_n \ket{n}$
\begin{equation}
\begin{aligned}
    \big|\langle\Psi|{\Upsilon}\rangle\big|^2 
    &= \Big| \sum^\nmax_{n=m} |c_n|^2 \mathrm{e}^{-i\Delta_k(n) \varphi }\Big|^2.
\end{aligned}
\end{equation}
Then using the expansion of the exponential to find
\begin{equation}
\begin{aligned}
    \big|\langle\Psi|{\Upsilon}\rangle\big|^2 
    &\ge 1 -\sum_{n=m}^\nmax|c_n|^2 [ \varphi \Delta_k(n)] ^2  -O(\varphi^4) \\
    &\ge 1 - \epsilon - O(\epsilon^2).
\end{aligned}
\end{equation}
Where, we used the assumption that $|\langle \Psi | \Psi \rangle|^2 = \sum_{n=m}^\nmax |c_n|^2 = 1$, then dropped all strictly positive terms, and lastly used $[ \varphi \Delta_k(n)] ^2 \leq \epsilon$ by assumption.
Therefore, if $\varphi$ is sufficiently small then $1 \geq |\langle \Psi | {\Upsilon} \rangle|^2 \geq 1 - O(\epsilon)$, and thus $|\langle \Psi | {\Upsilon} \rangle|^2 = 1 - O(\epsilon)$, independent of $\ket{\psi}$.
\end{proof}

From \cref{eq:go_no_go} we may read off the critical phase, $\varphi_c$:
\begin{equation}\label{eq:phiLim}
|\varphi| < |\varphi_c| \sim \gamma \frac{\sqrt{\epsilon}}{ | k\ g'(n)|},
\end{equation}
for all $n\in \{0,\ldots, \nmax\}$, and where $\gamma \geq 1$ is a scaling factor that depends on the specific probe state. 
The scaling factor diagnoses how much a specific state is affected by $\hat{V}_k(\varphi)$.
In the worst case scenario for any general probe state, proved by \cref{thrm:noise}, $\gamma \sim 1$. It may seem odd to bound an unknown parameter which we aim to estimate. It's natural for N00N states~\cite{N00N}, which can only sense phases less than $2\pi/N$.  The critical phase similarly limits the maximum interrogation time between correction steps in error-corrected metrology~\cite{suppMat,zhou2018}. Large evolution steps let amplified errors grow uncontrollably, restricting sensing to either small phases or rapid correction steps on the order of $1/g'(n)$~\cite{zhou2018}.

In general,~\cref{eq:phiLim} must be satisfied for each $k$, or the maximum $k$, corresponding to an $\E{k}$ in the error channel decomposition.
If~\cref{eq:phiLim} is always satisfied, then~\cref{eq:overlap} can be used to upper bound on the change in QFI:
\begin{equation}\label{eq:QFIDiff}
\big| I\big( \varphi| \hat G, \psi \big) - I\big( \varphi|\hat G,\Upsilon \big) \big| \approx\mathcal{O}(\sqrt{\epsilon}),
\end{equation}
where $I\big( \varphi| \hat G, \Upsilon \big)$ is the QFI for the case that the error occurs prior to the unitary encoding.
The derivation of this bound is provided in the SM~\cite{suppMat}. 

Surprisingly, if the $k$ dependence is heralded for each independent trial, through e.g. measurement, there is a small increase in the QFI. In that case $\hat V_k(\varphi)\E{k}\hat U(\varphi)$ can be viewed as the new encoding step. 
The resulting QFI, calculated in SM~\cite{suppMat}, is $\propto [g(n_{\rm max}+k)-g(n_{\rm min}+k)]^2$ for the optimal state, but it's no longer an asymptotic bound since $k$ may vary between trials.
Furthermore, treating this as a new encoding step requires that one knows \textit{precisely when} the error occurred. We have so far avoided this problem by commuting all errors to before the unitary encoding.
Otherwise, we would obtain $\hat V_k(\theta)\E{k}\hat U(\varphi)$, where $0 \leq |\theta| \leq |\varphi|$ depends on the error timing and introduces a nuisance parameter~\cite{SuzukiHayashi2020}.
The QFI for this case is shown in the SM~\cite{suppMat}, and is strictly less than the case with no-error.

For phases smaller than the critical phase, ~\cref{eq:QFIDiff} shows it is possible that the gains in sensitivity due to the nonlinear encoding may still outweigh the effect of increased noise, but the exact details of this trade-off depend on the probe state.
To investigate this we compute a fidelity, based on ~\cref{eq:overlap}, that illuminates how propagated errors corrupt the estimation procedure. Note that \cref{eq:overlap} only decreases due to $\hat{V}_k(\varphi)$, provided the assumption $|\langle \Psi | \Psi \rangle|^2 = 1$ holds. 
Therefore we can fully remove both $\E{k}$ and $\hat{U}_g(\varphi)$ to study just the effect of the emergent errors.
Then the fidelity becomes
\begin{equation}\label{eq:fidelity}
    F(k,g,\varphi,\psi) = |\langle\psi|\hat{V}_k(\varphi)|\psi\rangle|^2.
\end{equation}
Clearly if $F\approx 0$ then $V_k(\varphi)$ has a large effect on $\ket{\psi}$ whereas $F\approx 1$ indicates little-to-no effect.

The fidelity in \cref{eq:fidelity}, averaged over all pure states~\cite{Molmer2007,Nielsen2002} up to the Fock space cutoff, $\nmax$, is:
\begin{equation}\label{eqn:avg_fid}
   \overline{F}_\nmax(k,g,\varphi) = \frac{\nmax+ \Big|\sum_{n=0}^\nmax \mathrm{e}^{-i\varphi(g(n)-g(n-k))}\Big|^2}{\nmax(\nmax+1)} .
\end{equation}
In the large photon number cutoff limit, $\nmax\gg 1$,~\cref{eqn:avg_fid} exhibits two distinct regimes of behavior where it is either unity, or zero.
For a general function $g$, the crossover between these regimes, which we identify by $\overline{F}_\nmax(k,g,\varphi) = 1 - \epsilon$ for $1\gg\epsilon>1/D$, occurs when
\begin{equation}\label{eq:phi_crit_avg}
|\varphi| \approx \varphi^\mathrm{avg}_c \equiv \frac{ 2 \sqrt{\epsilon} }{ \sum_{n,m = 0}^D ( \Delta_k(n) - \Delta_k(m) )^2},
\end{equation}
where, as before, $\Delta_k(n) = g(n+k) - g(n)$.
In the limit that the photon number $\nmax$ tends towards infinity, this crossover becomes abrupt and discontinuous, resembling a phase transition between the regimes of usable and unusable non-linear quantum metrology.

 \begin{figure}[t]
    \centering
    \includegraphics[width=\columnwidth]{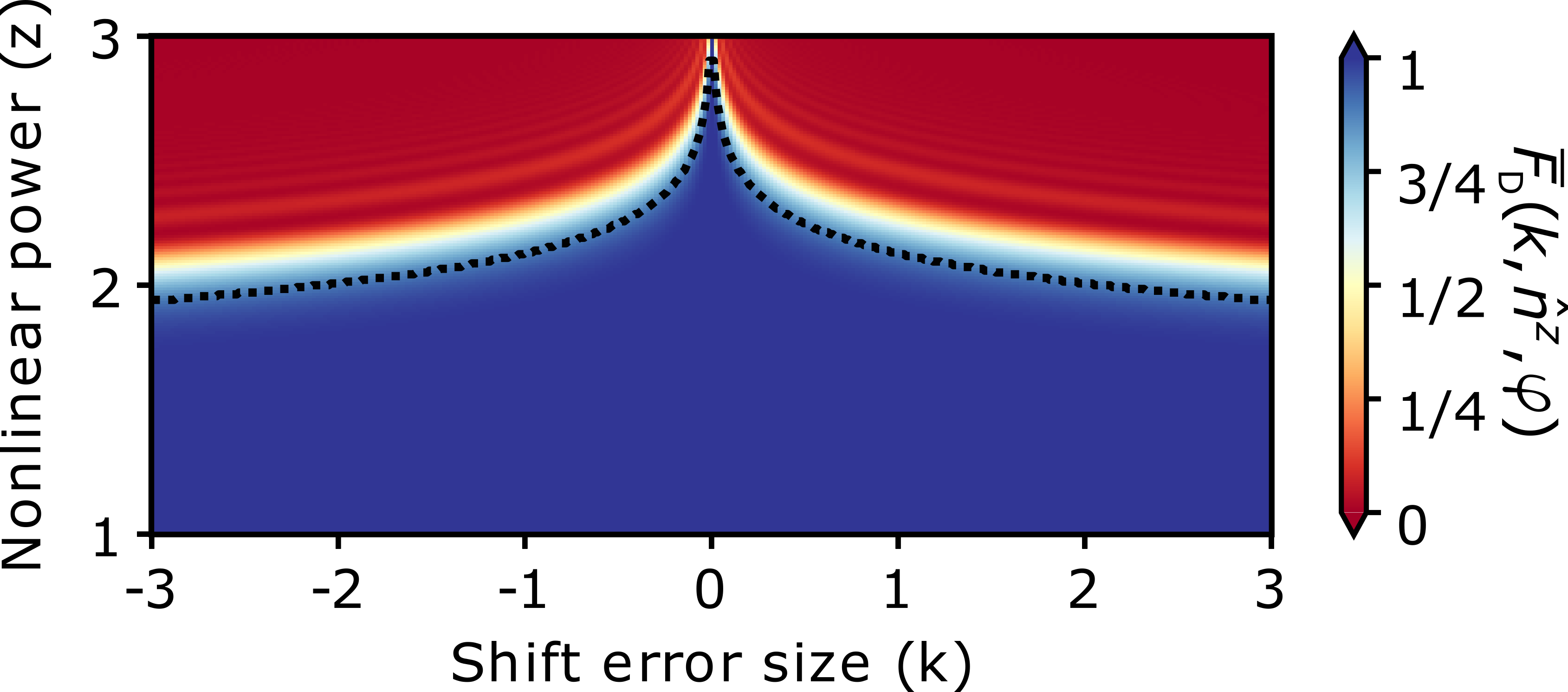}
    \caption{
    The crossover in metrological usefulness of noisy nonlinear metrology. The average fidelity metric \cref{eqn:avg_fid} is plotted as a function of the shift error $k$, analytically continued to the reals, and nonlinearity $z$ where $g(\hat{n}) = \hat{n}^{z}$. This plot is generated with a particular value of $\varphi = \pi/(10\nmax)$ and $\nmax=200$, increasing $\varphi$ will lower the nonlinear power where the phase transition occurs. As the size of shift error $k$ increases in magnitude, we see the emergence of a phase transition between fidelities near 1 and near 0. We plot the  curve where fidelity crosses $.9$ (dotted line), i.e. $\epsilon=0.1$.
    }
    \label{fig:transition}
\end{figure}

To understand this better, consider $g(\hat n) = \hat n^z$ for some power $z\ge 1$. This is a $z$-body interaction and sufficiently captures the scaling of any analytic $g(n)$.
For example, in the self-Kerr effect, $g(\hat n) = \hat n^2$, the limit of large photon number cutoff gives
\begin{equation}
    \overline{F}_\nmax(k,n^2) = \frac{\sin^2(\varphi \nmax k)}{(\varphi \nmax k)^2}.
\end{equation}
This will be near zero unless $k=0$, or $|\varphi| \propto 1/\nmax$ whereupon it will be one. 
For the average state,~\cref{eq:phi_crit_avg} suggests that a higher power of $z$ will further restrict the domain of $\varphi$. One could design the probe state to avoid this behavior altogether, but we leave that for future work.

By considering a fixed $\varphi$ but allowing the function $g$ to vary, one can study the crossover as a response to the nonlinearity of the function itself.
In~\cref{fig:transition}, we plot the average fidelity for the case that $g(\hat{n}) = \hat{n}^{z}$, as a function of $z$ and $k$. 
We choose a fixed value of $\varphi = \pi/(10\nmax)$, which, in the $\nmax\rightarrow \infty$ limit, corresponds to a discontinuous transition near $z=2$.
For all values of $z<2$, there is little to no loss of fidelity as $\nmax$ grows, while $z>2$ corresponds to a vanishing fidelity.

Lastly, we highlight an exotic method to avoid this critical behavior altogether via engineering the nonlinear Hamiltonian itself.
Since the critical value $\varphi_c$ depends on the derivative of $g$, one could engineer a function $g(\hat{n})$ that is locally flat around some operating point of the probe state but globally scales nonlinearly with photon number. 
For example, if one's goal is a nonlinear scaling of $g(\hat n)= \hat{n}^z$ for integer $z$, then one may construct a related function $\tilde{g} (\hat{n})$ that is locally flat, such as
\begin{equation}\label{eq:floorH}
    \tilde{g}(\hat{n}) = \left(\left \lfloor \frac{\hat{n}}{\mu}\right \rfloor \mu \right)^z.
\end{equation}
This function is a series of flat plateaus of length $\mu$, that scales globally as $n^z$. We choose a probe state to be a superposition of the state in the middle of the zeroth plateau and the state in the middle of the $p$'th plateau.
Such a probe state could endure shift errors of up to $\pm \mu/2$ and still be useful.
Highly engineered interactions of this form may be difficult to create, but could be implemented in certain platforms using, e.g. SNAP gates~\cite{SNAPgateEXP2015,SNAPgateTHY2015}.

%==========================================
\textit{Conclusions}. ---
%==========================================
In this work, we demonstrate the existence of emergent errors arising from the interplay between nonlinear metrology and noise. Our approach decomposes an error channel into an operator basis and analyzes the size of the emergent error for arbitrary basis elements. 
We identified a sufficient condition--the critical phase--where noise has negligible impact on nonlinear metrology. The critical phase is a value of the parameter to be sensed below which the noise doesn't harm nonlinear metrology. The critical phase primarily depends on the derivative of the nonlinear function of $\hat n$ i.e. $\phi_c\sim 1/ g'(\hat n)$.  
This dependence, $ 1/g'(\hat{n})$, also determines the time interval between error correction steps in error-corrected metrology~\cite{zhou2018}, making it a crucial factor for advanced quantum metrology schemes. 
Lastly, we highlight that leveraging heralded information could preserve nonlinear enhancements in the presence of noise and thus open new paths for advancing quantum metrology.

In the past, there have been arguments about whether highly nonlinear metrological schemes are physical~\cite{Roy2008,Boixo2007,Boixo2008}. However, modern experiments are rapidly creating the opportunity for these nonlinear systems to be studied and leveraged for quantum advantage~\cite{Beau2017,Tsarev_2019,nie2018,luo2024}.
The existence of emergent errors does not undo the possible benefits of nonlinear metrology; rather it highlights an important consideration in achieving this advantage. \\

\noindent {\em Acknowledgments:} The authors acknowledge helpful discussions with Ivan H Deutsch and Michael J W Hall. All authors were supported by National Science Foundation under QLCI Award No. OMA-2016244. NL and JC also acknowledge support from the Office of Naval Research Award No. N00014-21-1-2606. JW and MH also acknowledge support from the Department of Energy Quantum Systems Accelerator, Grant No. 7565477, and the National Science Foundation, Grant No. PHY-2317149.

\newpage 

\appendix
\begin{widetext}

\section{Supplemental Material: Noise constraints on sensitivity scaling in quantum nonlinear metrology}

\tableofcontents

\section{Explicit Calculations of the Phase Error}
The emergent phase error is due to the commutation of $\hat{U}_g(\varphi)$ and $\Ep{k}{\Phi}$, and can be implicitly defined by
\begin{equation}
\hat{U}_g(\varphi)\Ep{k}{\Phi} = \hat{V}_k(\varphi)\ \Ep{k}{\Phi} \hat{U}_g(\varphi).
\end{equation}
To calculate $\hat{V}_k(\varphi)$ we will do it for the two casses in the definition of $\Ep{k}{\Phi}$:
\begin{equation}
 \Ep{k}{\Phi} \equiv \left\{ \begin{array}{ll}
\mathrm{e}^{i\Phi \hat{n}} \Sigma_{|k|}^- & \text{ for } k < 0 \\
\Sigma_{|k|}^+ \mathrm{e}^{-i\Phi \hat{n}} & \text{ for } k \ge 0
  \end{array}\right. .
\end{equation}

In order to calculate $\hat{V}_k(\varphi)$ it is simplest to do for two cases, first when $k\geq 0$ and then when $k<0$. In both cases, we will show that $\hat{V}_k(\varphi)$ is the same.
First, assume that $k\geq0$,
\begin{equation}
\begin{aligned}
\hat{U}_g(\varphi)\Ep{k}{\Phi} =& \left( \sum_{m\in \mathbb S} \mathrm{e}^{ -i \varphi g(m) } \op{m}{m} \right)  \sum_{n\in \mathbb S} \op{n+k}{n} \mathrm{e}^{- i \Phi n } \\
\end{aligned}
\end{equation}
whereupon, if $n+k \notin \mathbb{S}$ then $\ket{n+k}$ is set to the null vector, i.e. $\ket{n+k} = 0$.
This case that occurs in, for example, the spin ensemble Hilbert space where you cannot excite a spin ensemble ``past'' all $N$ spins being excited.
We comment further on this at the end of the calculation.
We have that use $m = n + k$ and
\begin{equation}
\begin{aligned}
\hat{U}_g(\varphi)\Ep{k}{\Phi} =& \sum_{n\in \mathbb S} \mathrm{e}^{ -i \varphi g(n+k) } \op{n+k}{n} \mathrm{e}^{- i \Phi n } \\
=& \sum_{n\in \mathbb S} \mathrm{e}^{ -i \varphi [ g(n+k) - g(n) ] } \op{n+k}{n}  \mathrm{e}^{- i \Phi n } \mathrm{e}^{ -i \varphi g(n) } \\
=& \mathrm{e}^{ -i \varphi [ g(\hat{n}) - g(\hat{n}-k) ] } \sum_{n\in \mathbb S} \op{n+k}{n}  \mathrm{e}^{- i \Phi n } \mathrm{e}^{ -i \varphi g(n) }
\end{aligned}
\end{equation}
where we've used the fact that $\mathrm{e}^{ -i \varphi [ g(\hat{n}) - g(\hat{n}-k) ] } \ket{m+k} = \mathrm{e}^{ -i \varphi [ g(n+k) - g(n) ] } \ket{m+k}$ to factor out the phase shift, which we now identify as $\hat{V}_k(\varphi)$ and finish the calculation:
\begin{equation}
\begin{aligned}
\hat{U}_g(\varphi)\Ep{k}{\Phi}  =& \hat{V}_k(\varphi) \left( \sum_{n\in \mathbb S} \op{n+k}{n}  \mathrm{e}^{- i \Phi n } \right) \left( \sum_{m\in \mathbb S} \op{m}{m}  \mathrm{e}^{- i \varphi g(m) } \right)  \\
=&  \hat{V}_k(\varphi)\ \Ep{k}{\Phi} \hat{U}_g(\varphi)
\end{aligned}
\end{equation}
as needed. 

Now, if $k<0$, we can do the analogous calculation:
\begin{equation}
\begin{aligned}
\hat{U}_g(\varphi)\Ep{k}{\Phi} =& \left( \sum_{m\in \mathbb S} \mathrm{e}^{ -i \varphi g(m) } \op{m}{m} \right)  \sum_{n\in \mathbb S} \mathrm{e}^{i \Phi (n+k) } \op{n+k}{n} \\
\end{aligned}
\end{equation}
whereupon, if $n+k \notin \mathbb{S}$ then $\ket{n+k}$ is once again set to the null vector.
This might be the case for both spin ensembles or bosonic ensembles, where negative photon numbers aren't considered physical.
Once again, we use $m = n + k$ and
\begin{equation}
\begin{aligned}
\hat{U}_g(\varphi)\Ep{k}{\Phi} =& \sum_{n\in \mathbb S} \mathrm{e}^{ -i \varphi g(n+k) } \mathrm{e}^{+ i \Phi (n+k) } \op{n+k}{n} \\
=& \hat{V}_k(\varphi) \left( \sum_{n\in \mathbb S} \mathrm{e}^{+ i \Phi (n+k) } \op{n+k}{n} \right) \left( \sum_{m\in \mathbb S} \op{m}{m}  \mathrm{e}^{- i \varphi g(m) } \right)  \\
=&  \hat{V}_k(\varphi)\ \Ep{k}{\Phi} \hat{U}_g(\varphi)
\end{aligned}
\end{equation}
as needed.
Therefore
\begin{equation}\label{eq:Vk}
 \hat{V}_k(\varphi) \equiv \mathrm{e}^{ -i \varphi [ g(\hat{n}) - g(\hat{n}-k) ] }.
\end{equation}

In the cases that $n+k \notin \mathbb{S}$, we have included extra states as an accounting trick.
This is equivalent to working in a lager Hilbert space to ``make $\Ep{k}{\Phi}$ unitary" (which it happens to be for $\mathbb{S} = \mathbb{Z}$) and then subsequently projecting back down to only the Hilbert space where we want to consider observables.
This projection could be handled by, for example, projection operators acting from the left to delete states not in the set being considered which is not unlike schemes for regularization~\cite{pauli1949}.
However, the space we would project onto is dependent on $\mathbb{S}$, thereby making these methods messier.
We highlight that any intermediate states which are un-physical, like $n+k \notin \mathbb{S}$, do not contribute to any observable physics in the final calculation because $\Ep{\theta}{k}$ would remove them by projection.

\section{Continuous Case}\label{sec:cont}
In the main text and in the calculation above, we have dealt exclusively with states labeled by discrete quantum numbers.
Here, we outline the case that one considers continuous operators and states which may be labeled by a continuum of quantum numbers.
All the results in the main text still hold without modification.

In the continuous Hilbert space, $\mathbb{S} = \mathbb{R}$, and we use the position operator $\hat x = \int_{-\infty}^{+\infty} \op{x}{x} dx$ as the generalized ``number'' operator, which has a continuous spectrum over the continuum of states $\ket{x}$ for $x \in \mathbb{S}$.
The non-linear Hamiltonian is now a function of $\hat x$ i.e.  $\hat{G} \equiv g(\hat{x})$.
This generates the unitary $U_g(\varphi) = \exp(-i\varphi \hat{ G } )$ where $\varphi \in \mathbb{R}$ need not be between $0$ and $ 2 \pi$. 
The analogous choice for the error basis is the displacement operators:
\begin{equation}
\left\{ \hat{D}_k(\Phi) = \hat{D}\left( \frac{ k \mathrm{e}^{-i \Phi}}{ \sqrt{2} } \right)\Bigg| k\in \mathbb{R}, \Phi \in(-\pi,\pi)\right\}
\end{equation}
where $D(\alpha) = \exp( \alpha \hat{a}^\dagger - \alpha^* \hat{a} )$ for $\alpha\in\mathbb{C}$ and $\hat{a} = ( \hat{x} + i \hat{p} ) / \sqrt{2}$.
Now, $k$ is the ``number'' of the shift error while $\Phi \in [0, \pi)$ decides the direction of the displacement, rather than adding a rotation.

It is straight forward to show the error propagation expression analogous to \cref{eq:Vk} is
\begin{equation}
\hat{U}_g(\varphi) \hat{D}_k(\Phi) =\hat{\mathcal{V}}_k(\varphi) \hat{D}_k(\Phi) U_g(\varphi).
\end{equation} 
Now, all the displacement errors are unitary so it is straightforward to solve:
\begin{equation}
\begin{aligned}
\hat{\mathcal{V}}_k(\varphi)=& \hat{U}_g(\varphi) \hat{D}_k(\Phi) \hat{U}^\dagger_g(\varphi) \hat{D}^\dagger_k(\Phi).
\end{aligned}
\end{equation}
Analogous to phase errors above, any errors corresponding to displacements with an $\hat{x}$ operator ($\Phi = \pi/2$) will commute with the Hamiltonian. 
Errors corresponding to displacements with a $\hat{p}$ operator will not commute, so we restrict to $\Phi = 0$. 
In this case, the propagation of this error through $U_g(\gamma)$ leads to the same nonlinear enhancement as every other case:
\begin{equation}
\hat{\mathcal{V}}_k(\varphi) = \mathrm{e}^{-i \varphi [ g(\hat{x}) - g(\hat{x} - k ) ] }.
\end{equation}

\section{The Go-No-Go-Law for other Hilbert Spaces}
Here, we briefly discuss the cases outside just that of the bosonic cases.
In the case of a spin ensemble, where $\mathbb{S} = \{0,...,N\}$ the proof goes the same where $M=N$ directly.
In the case of the rigid rotor, where $\mathbb{S} = \mathbb{Z}$ we find one minor modification.
We again consider a cutoff, but now it is symmetric about zero so that we restrict to $-M\leq n \leq M$.
This means
\begin{equation}
\E{k}{}^{\dagger} \hat{V}_k(\varphi) \E{k} = \sum_{n=-M}^M \mathrm{e}^{-i \Delta_k(n) \varphi } \op{n}{n},
\end{equation}
while the rest of the proof continues the same, with the added nicety that $\E{k}$ is now unitary.

In a continuous variable Hilbert space, where $\mathbb{S} = \mathbb{R}$, we will find that the cutoff, $M$, should instead be a maximum cut-off in the ``position'' operator.
We outline the modifications to the formalism in Sec.~\ref{sec:cont}.
The proof of the Go-No-Go-Law remains the same, except all operators are now unitary and the sums over $\mathbb{S}$ are integrals.
The critical phase may now be interpreted as a critical displacement below which the non-linear errors won't have a significant effect.
For posterity, the full proof is shown in Sec.~\ref{sec:cont} as well.

\subsection{Critical ``Phase'' in The Continuous Case:}
Now, we have the tools to understand the continuous case of the Go-No-Go-Law, which largely goes the same as the discrete cases.
Here, we will work in the ``position'' basis and, rather than a max photon cut-off, we will consider a maximum position, $A$.
First, for $\Phi \neq 0$ we note that $\hat{D}_k(\Phi)$ has a component which commutes with $g(\hat{x})$.
To simplify the discussion, in analogy to the main text, we will fix $\Phi = 0$;
\begin{equation}
\hat{D}_k = \exp\left( \frac{k}{\sqrt{2}} \left( \hat{a}^\dagger - \hat{a} \right) \right) = \exp( - i k \hat{p} ).
\end{equation}
Now, we turn to the Go-No-Go-Law, which we restate here.

\emph{The Go-No-Go-Law:} For a fixed shift $k$, if the encoded displacement $\varphi$ is sufficiently small such that 
\begin{equation}
|\varphi| \leq \frac{\sqrt{\epsilon}}{  \max_{|x| \leq A} | g(x+k) - g(x) | },
\end{equation}
with $\epsilon \ll 1$, then $|\langle \Psi | \tilde{\Psi} \rangle|^2 = 1 - \mathcal{O}( \epsilon)$.

Proof: First, let $\Delta_k(x) \equiv g(x+k) - g(x)$. 
Now, take $k$ to be fixed, and assume $|\varphi| \leq \sqrt{\epsilon} / ( \max_{|x| \leq A} |\Delta_k(x) | ) $ for $\epsilon \ll 1$. Notice that
\begin{equation}
\E{k}{}^{\dagger} \hat{\mathcal{V}}_k(\varphi) \E{k} = \int_{-A+k}^{A+k} e^{-i \varphi \Delta_k(x)} \op{x}{x} dx.
\end{equation}
Now, expand the probe state in the spatial basis, yielding an integral rather than a sum; $\ket{\psi} = \int_{-A+k}^{A+k} c(x) \ket{x} dx$.
This means
\begin{equation}
\begin{aligned}
    \big|\langle\Psi|\tilde{\Psi}\rangle\big|^2
    &= \Big| \int_{-A+k}^{A+k} |c(x)|^2 \mathrm{e}^{-i\Delta_k(n) \varphi }\Big|^2 dx.
\end{aligned}
\end{equation}
Then using the expansion of the exponential to find
\begin{equation}
\begin{aligned}
 \big|\langle\Psi|\tilde{\Psi}\rangle\big|^2 
    &\ge 1 -\int_{-A+k}^{A+k} |c(x)|^2 [ \varphi \Delta_k(x)] ^2  -\mathcal{O}(\varphi^4) dx \\
    &\ge 1 - \epsilon - \mathcal{O}(\epsilon^2).
\end{aligned}
\end{equation}
Here, we used the assumption that $|\langle \Psi | \Psi \rangle|^2 = \sum_{n=m}^M |c(x)|^2 = 1$, then dropped all strictly positive terms, and lastly used $[ \varphi \Delta_k(x)] ^2 \leq \epsilon$ by assumption.
Therefore, if $\varphi$ is sufficiently small then $1 \geq |\langle \Psi | \tilde{\Psi} \rangle|^2 \geq 1 - \mathcal{O}(\epsilon)$, and thus $|\langle \Psi | \tilde{\Psi} \rangle|^2 = 1 - \mathcal{O}(\epsilon)$, independent of $\ket{\psi}$. $\blacksquare$

From \cref{eq:go_no_go} we see a scaling law: non-linear metrology may be useful so long as the phase, $\varphi$, is smaller than
\begin{equation}
|\varphi| < |\varphi_c| \sim \gamma \frac{\sqrt{\epsilon}}{ | k\ g'(n)|},
\end{equation}
for all basis states $\ket{n}$ in the support of the probe state, and where $\gamma$ depends on the specific probe state. 
In the worst case scenario, proved by the Go-No-Go-Law, $\gamma \sim 1$.

\section{Perturbative QFI Expansion}
When we have that
\begin{equation}
|\langle \Psi | {\Upsilon} \rangle|^2 = 1 - O( \epsilon)
\end{equation}
We can expand the state
\begin{equation}
\ket{\Upsilon} \approx \sqrt{1-\epsilon}\ket{\Psi} - \sqrt{\epsilon} \ket{\perp}
\end{equation}
where $|\perp\rangle$ can be obtained from a Gram Schmidt procedure so that
\begin{equation}
\langle \perp | \Psi \rangle = 0.
\end{equation}

From this expansion, we can perturbatively compute the QFI, where we assume that the state $\ket{\perp}$ carries no usable information about $\varphi$, i.e. we take $ \frac{\partial}{\partial \varphi} \ket{\perp} = 0$, when we calculate the QFI.
From here on out, we will use the shorthand $\partial_\varphi = \frac{\partial}{\partial \varphi}$.

Here, we will calculate the QFI using the derivative definition:
\begin{equation}
\begin{aligned}
I( \varphi| \Upsilon) =& 4 \langle \partial_\varphi \Upsilon | \partial_\varphi \Upsilon \rangle - 4 \langle \partial_\varphi \Upsilon | \Upsilon \rangle \langle\Upsilon | \partial_\varphi  \Upsilon \rangle
\end{aligned}
\end{equation}
where, using the approximation,
\begin{equation}
| \partial_\varphi  \Upsilon \rangle = \sqrt{1 - \epsilon} \ket{\partial_\varphi \Psi }.
\end{equation}
This means
\begin{equation}
\begin{aligned}
I( \varphi| \Upsilon) =& 4 (1-\epsilon) \langle \partial_\varphi \Psi | \partial_\varphi \Psi \rangle - 4 (1-\epsilon) \langle \partial_\varphi \Psi | \Upsilon \rangle \langle\Upsilon | \partial_\varphi  \Psi \rangle \\
=& 4 \langle \partial_\varphi \Psi | \partial_\varphi \Psi \rangle - 4 \langle \partial_\varphi \Psi | \Upsilon \rangle \langle\Upsilon | \partial_\varphi  \Psi \rangle + \mathcal{O}(\epsilon) \\
=& 4 \langle \partial_\varphi \Psi | \partial_\varphi \Psi \rangle - 4 \langle \partial_\varphi \Psi | \Psi \rangle \langle \Psi | \partial_\varphi  \Psi \rangle \\
& - 4 \sqrt{\epsilon} \sqrt{1-\epsilon} \langle \partial_\varphi \Psi | \Psi \rangle \langle\perp | \partial_\varphi  \Psi \rangle - 4 \sqrt{\epsilon} \sqrt{1-\epsilon} \langle \partial_\varphi \Psi | \perp \rangle \langle \Psi | \partial_\varphi  \Psi \rangle + \mathcal{O}(\epsilon) \\
=& 4 \langle \partial_\varphi \Psi | \partial_\varphi \Psi \rangle - 4 \langle \partial_\varphi \Psi | \Psi \rangle \langle \Psi | \partial_\varphi  \Psi \rangle + \mathcal{O}(\sqrt{\epsilon})
\end{aligned}
\end{equation}
whereupon we can identify this first term as $I(\varphi|\Psi)$.
Now, we note that
\begin{equation}
\ket{\Psi} = \E{k} U_g(\varphi) \ket{\psi},
\end{equation}
where, so long as $\E{k}$ does not decrease the norm of the state, we can use the fact that any actions taken after the unitary encoding of $\varphi$ do not modify the QFI with respect to $\varphi$.
Therefore, $I(\varphi|\Psi) = I(\varphi|\hat G,\psi)$ and we can conclude that
\begin{equation}
\big| I\big( \varphi| \hat G, \psi \big) - I\big( \varphi|\hat G,\Upsilon \big) \big| \approx\mathcal{O}(\sqrt{\epsilon})\,.
\end{equation}

\section{Heralding The Shift Dependence}

As mentioned in the main text, one could consider heralding the $k$ dependence and take that into account.
Each independent trial that culminates in a measurement will have a different value of $k$ which, in principle, depends on the error channel itself.
This means that modifications to the quantum Cram\'er Rao bound (and therefore the asymptotic sensitivity) depend on the specific error channel being considered.
However, we can still calculate the QFI for one of these trials in which $k$ has been heralded.
Here, we will do this two ways.

Throughout the main results, we assume that the error occurs prior to the beginning of the unitary encoding.
We can do this without a loss of generality because we may commute all errors from the state preparation, encoding, and measurement to just before the ideal encoding step. 
We subsequently make the simplifying assumption that $\Ep{k}{\Phi} \rightarrow \E{k}$ with $\Phi = 0$.
However, this phase error contains the accumulated phase from pulling the error through the phase encoding, initially.

First, we will continue to ignore this initially accumulated phase and calculate the QFI in the case that one the heralds $k$ dependence.
Second, we will include the fact that the error might have occurred \textit{during} the phase encoding, and show that the time at which the error occurs becomes a nuisance parameter and turns the problem into a multi-parameter estimation problem.

\subsection{Errors Prior to Encoding}
Calculating the QFI for the simple case, in which we continue operating under the simplifying assumption that the error occurs before the unitary encoding step , means we must calculate the QFI for the state
\begin{equation}
\ket{\Upsilon} = \hat{V}_k(\varphi) \E{k} \hat{U}_g(\varphi) \ket{\psi}
\end{equation}
where we know the value of $k$ and wish to estimate $\varphi$.
First, we can make use of the fact that any actions we take after the value of $\varphi$ has been encoded onto the state will not change the QFI.
This means we can instead calculate the QFI for \begin{equation}
\begin{aligned}
\ket{\tilde{\Upsilon}} &= \E{k}^\dagger \hat{V}_k(\varphi) \E{k} \hat{U}_g(\varphi) \ket{\psi} \\
&= \left( \sum_{n=m}^\infty \mathrm{e}^{-i g(n+k) \varphi } \op{n}{n} \right)\ket{\psi}
\end{aligned}
\end{equation}
for $m = \mathrm{max}(-k,0)$.
From this, we can conclude that the QFI is
\begin{equation}
\begin{aligned} \label{eq:shiftedQFI}
I(\varphi|g(\hat{n}),\Upsilon) =& 4 \  \var_\psi(g(\hat{n}+k)),
\end{aligned}
\end{equation}
where $\var_\psi(\hat{A}) = \langle \psi| \hat{A}^2 | \psi \rangle - \langle \psi| \hat{A} | \psi \rangle^2$ is the variance.
This is the same QFI, but now for the generator shifted by $k$, which is expected. When using the optimal input state, i.e. $\ket{\psi_{\rm opt}} \equiv (\ket{n_{\rm min}}+\ket{n_{\rm max}})/\sqrt{2}$ where   $n_{\rm max}> n_{\rm min}$, we can directly evaluate~\cref{eq:shiftedQFI} to find  $\propto [g(n_{\rm max}+k)-g(n_{\rm min}+k)]^2$.

\newpage

\subsection{Errors During Encoding}

Since errors accumulate continuously in time, to herald the $k$ dependence we must account for the fact that the time at which the error occurs matters.
If we did this for a specific error channel, this would be accounted for by $\Phi$ in the error basis, with $\Ep{k}{\Phi}$.
Here, we can explore this dependence in general with the simple circuit model:
\begin{equation}
\resizebox{0.5\textwidth}{!}{
\raisebox{0.25em}{
\begin{quantikz}[row sep=0.295cm,column sep=0.2cm]
    & \gate[][0.7cm][0.7cm]{\hat{U}_g(\varphi-\theta)} & \gate[][0.7cm][0.7cm]{\E{k}} & \gate[][0.7cm][0.7cm]{\hat{U}_g(\theta)}& \qw
\end{quantikz}
\hspace{-0.1 em} \raisebox{-0.25em}{=}
\begin{quantikz}[row sep=0.295cm,column sep=0.2cm]
    & \gate[][0.7cm][0.7cm]{\hat{U}_g(\varphi) } & \gate{\E{k}} & \gate{\hat{V}_k(\theta)} & \qw
\end{quantikz}
},\tag{7a}
}
\end{equation}
where $\theta$ is a parameter between zero and $\varphi$ such that $0 \leq |\theta| \leq |\varphi|$.
In words, $\theta$ captures the dependence on the fact that the error might have occurred \textit{during} the unitary phase encoding.
Here, $\hat{V}_k(\theta) = \exp(-i \theta \left[ g(\hat{n}) - g(\hat{n} - k) \right] )$ is exactly the same.

This means we want to calculate the QFI matrix for $\varphi$ and $\theta$ and the state
\begin{equation}
\ket{\psi( \varphi,\theta) } =  \hat{V}_k(\theta) \E{k} \hat{U}_g(\varphi) \ket{\psi}
\end{equation}
Now, we want to calculate the QFI matrix:
\begin{equation}
\mathbf{I} = \begin{pmatrix}
 I_{\varphi,\varphi} & I_{\varphi,\theta} \\
 I_{\varphi,\theta} & I_{\theta,\theta}
\end{pmatrix}
\end{equation}
where $I_{\varphi,\varphi}$ and $I_{\theta,\theta}$ are the QFI's for each parameter, while $I_{\varphi,\theta}$ are the correlations between the parameters.
The quantum Cram\'er Rao bound is then
\begin{equation}
\Delta \phi^2 \geq \left[ \nu \left( I_{\varphi,\varphi} - \frac{ I_{\varphi,\theta}^2 }{ I_{\theta,\theta} } \right) \right]^{-1} .
\end{equation}
where, basically, $ I_{\varphi,\theta}^2 / I_{\theta,\theta} $ accounts for the correlations which degrade the QFI with respect to $\varphi$.

Using the same principle that any actions we take after the value of $\varphi$ has been encoded onto the state will not change the QFI, we know that 
\begin{equation}
I_{\varphi,\varphi} = 4 \var_\psi(g(\hat{n}))
\end{equation}

For the QFI with respect to $\theta$, however, we'll let $\ket{p} = \E{k} \hat{U}_g(\varphi) | \psi \rangle $ for brevity, so that
\begin{equation}
\begin{aligned}
I_{\theta,\theta} =& 4 \var_p( g(\hat{n}) - g(\hat{n} - k) )
\end{aligned}
\end{equation}
where we will use that
\begin{equation}
\bra{p} \left( g(\hat{n}) - g(\hat{n} - k) \right)^a \ket{p} = \bra{\psi} \left( g(\hat{n}+k) - g(\hat{n}) \right) ^a \ket{\psi},
\end{equation}
for $a = 1,2$, to conclude
\begin{equation}
\begin{aligned}
I_{\theta,\theta} =& 4 \var_\psi( g(\hat{n}+k) - g(\hat{n}) )
\end{aligned}
\end{equation}
From these two, or a direct calculation using the symmetric logarithmic derivative, we may also find that
\begin{equation}
I_{\theta,\varphi} = 4 \mathrm{cov}_\psi( g(\hat{n}), g(\hat{n}+k) - g(\hat{n}) )
\end{equation}
where
\begin{equation}
\begin{aligned}
\mathrm{cov}_\psi( \hat{A}, \hat{B} ) = \langle \psi| \frac{ \hat{A} \hat{B} + \hat{B} \hat{A} }{2} | \psi \rangle - \langle \psi| \hat{A} | \psi \rangle \langle \psi| \hat{B} | \psi \rangle. 
\end{aligned}
\end{equation}
We can interpret $I_{\varphi,\varphi} - \frac{ I_{\varphi,\theta}^2 }{ I_{\theta,\theta} } = J(\varphi|\psi)$ as the effective QFI in the presence of the nuisance parameter $\theta$.
Notably,
\begin{equation}
J(\varphi|\psi) \leq I(\varphi|g(\hat{n}), \psi) .
\end{equation}

\begin{figure}[h]
    \centering
    \includegraphics[width=0.55\linewidth]{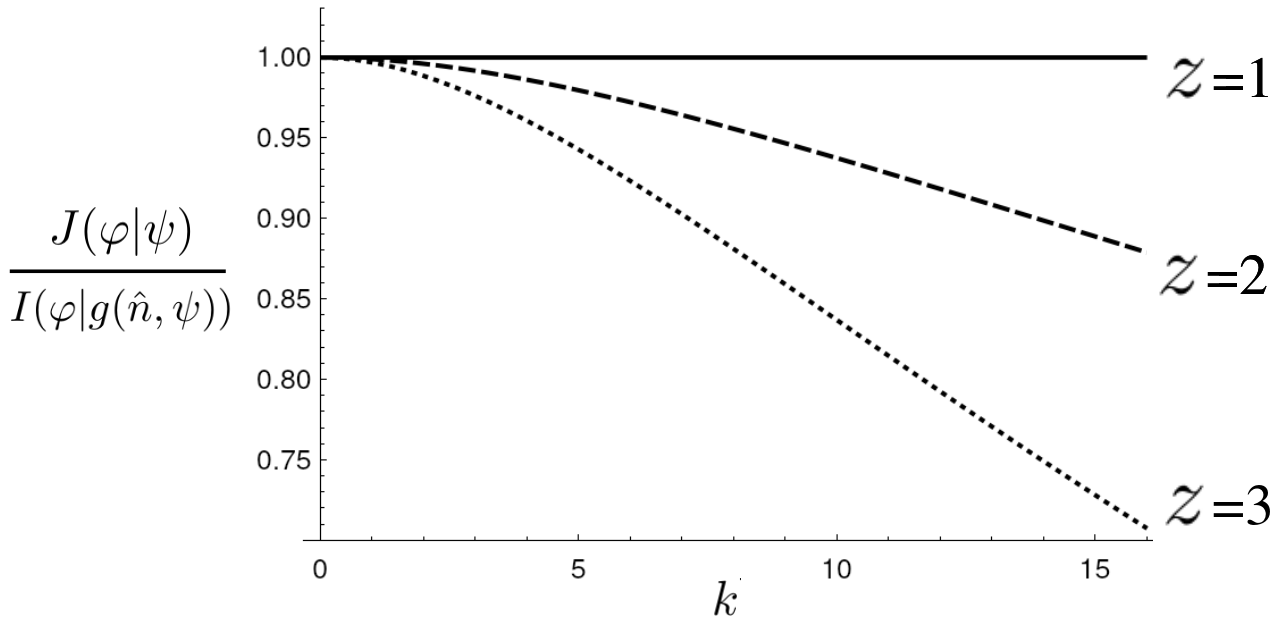}
    \caption{The degradation of the QFI with an M-N state for $M = 10$, $N=50$.}
    \label{fig:Fig1}
\end{figure}

For $\ket{\psi} = \frac{1}{\sqrt{2}} ( \ket{M} + \ket{N} )$ in Fock space, i.e. an $M-N$ state, we find that
\begin{equation}
J(\varphi|\psi) = \frac{ (g[N] - g[M])^3 }{ ( g[N + k] - g[M+k] )^2 } \ (g[N] - 2 g[N + k] - g[M] + 2 g[M+k])
\end{equation}
where we can use $g(\hat{n}) = \hat{n}^z$ to construct~\cref{fig:Fig1}.

\section{Interpreting critical parameter value with error correction. }

Here we consider a phase encoding where the parameter being encoded is $\phi = \omega T$ and $U_g(\phi) = \int_0^T \exp(-i\omega t g(\hat{n})) $. We assume that the probe state is error corrected with error correction steps occurring after intervals of length $\Delta t$. For simplicity we will consider the case where error correction happens instantly, or alternatively where error correction steps are fast enough that negligible phase is encoded during them, and errors do not occur during them. We can now decompose the action of our unitary into each interval of length $\Delta t$ 
    \begin{equation}
        U_g^{(i)}(\phi) = \int_{t_i}^{t_i + \Delta t} \exp(-i\omega t g(\hat{n}))
    \end{equation}
    where the superscript $i$ denotes the $i^{th}$ interval. We can now consider a shift error $\mathbb{E}_k(0)$ occurring during the $i^{th}$ interval at some time $t_e\in (t_i ,t_i+\Delta t)$. We know that this will induce an effective phase error due to the non-linearity, but after the shift is corrected the phase error will not grow in size. The worst case occurs when the error occurs at the beginning of the interval $t_e = t_i$. In this case, phase will be applied to the error state for a time $\Delta t$ before the state is corrected. This results in a phase error $\hat{V}_k(\varphi = \omega\Delta t)$.

    These errors can be mitigated either by reducing the time between the error correcting steps $\Delta t$, or by reducing the physical error rate. Of course, doing this is highly nontrivial because the actual error-correcting steps often take significant time to perform. As the percent of time spent error-correcting increases, there will be a constant factor reduction in the scaling of the phase encoding. So long as this constant factor is not large enough to undue the improved scaling, then the nonlinearity may still be useful for metrology. 

\section{Explicit Error Compensation}

\begin{figure}[h!]
  \centering
  \includegraphics[width=\textwidth]{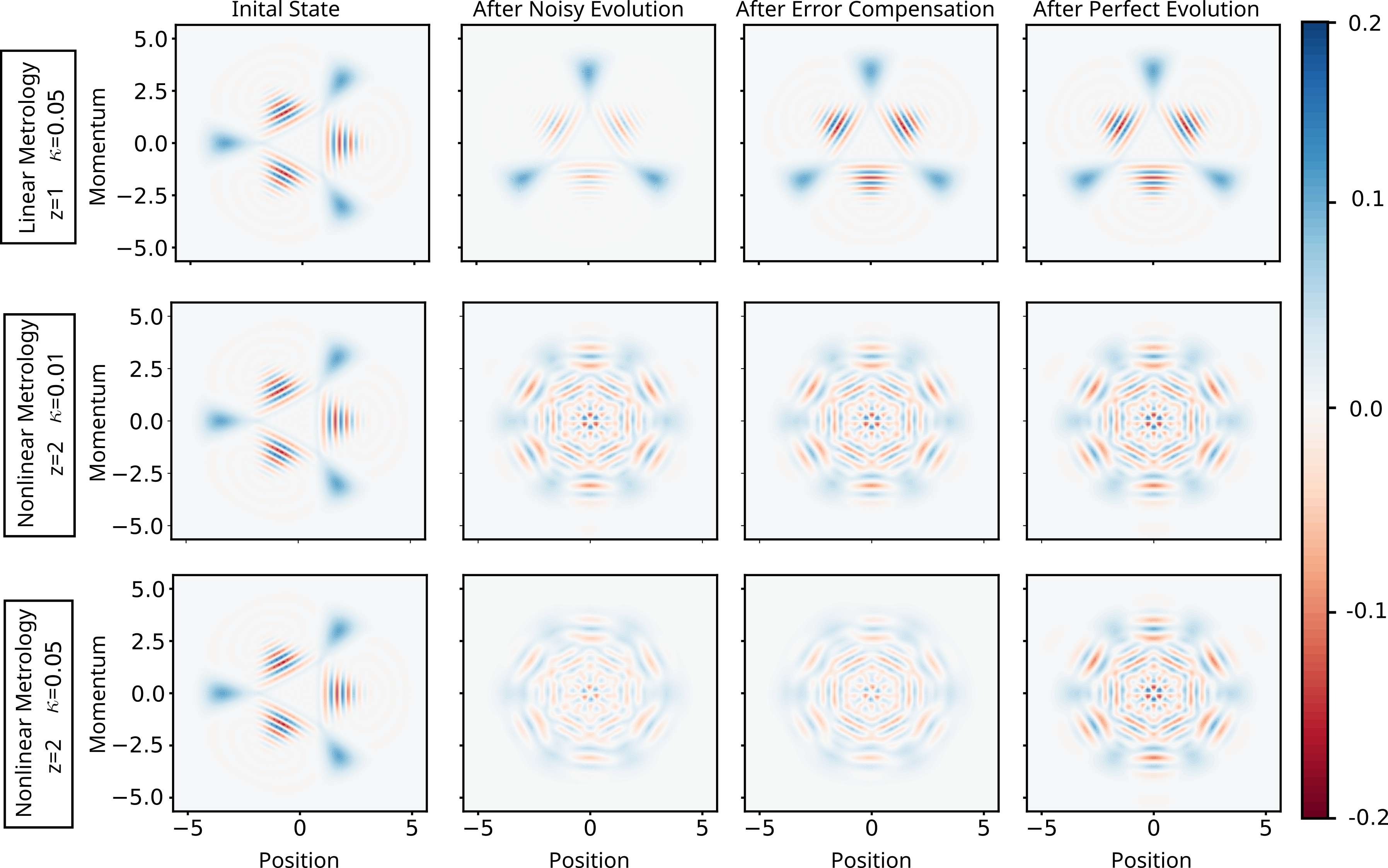}
  \caption{Wigner function of a binomial encoding through various steps of noisy linear metrology in the top panel and noisy nonlinear metrology in the bottom panel. The first column presents the initial probe state prepared in the minus state of a binomial code with a three fold rotation symmetry and an average photon number of 12. This probe state is evolved via a Lindblad master equation approach via a hamiltonian $\omega\hat{n}$ ($\omega \hat{n}^2$) in the top (bottom) panel. Here, $t = 1$ is fixed and $\omega t = \varphi$, where $\omega$ is varied to sweep values of $\varphi$. The master equation also includes photon loss with a unitless loss rate $\kappa$. The state after this evolution is pictured in the second column. In the third column we apply a photon number stabilizer and a recovery shift to compensate for any photon loss errors~\cite{Marinoff2024}. We can compare this state to the state in the fourth column which evolved under the same hamiltonian, but with no loss, $\kappa = 0$. Row one is the case of linear metrology, with significant loss. Correcting the loss errors is sufficient to reproduce the perfect state with high fidelity. Row two shows nonlinear case with low loss such that the state is minimally perturbed. The third row shows the nonlinear case with significant loss such that the state can not be recovered with hight fidelity.    }
  \label{fig:Correcting code}
\end{figure}

We have considered the case where the probe state suffers an explicit shift error $\mathbb{E}_k$ and is able to compensate for this error exactly. Explicitly, we assume that the probe state is not in the kernel of the error to ensure that $\mathbb{E}_k^\dagger$ acts as the recovery for our probe state. In this simple case, we have omitted several details such as the specifics of the error channel, how the probe state is encoded, and what stabilizers are measured. 

In this section, we implement error-correcting code in a bosonic oscillator system. We then subject the encoded state to a free evolution under a nonlinear Hamiltonian $\hat H$ with Linbladian loss given by some loss rate $\kappa$. That is we evolve under the master equation
\begin{equation}\label{eq:lossme}
    \dot {\hat{\rho}} = -i [\hat H,\hat \rho] +\kappa \mathcal D[\hat a] \hat\rho ,
\end{equation}
integrated up to some unitless time $\kappa t$,
with $\mathcal D[\hat L] \hat\rho = \hat L \hat\rho \hat L\dg - \half \hat L\dg \hat L \hat\rho - \half \hat\rho \hat L\dg \hat L$. After applying a set amount of phase $\varphi$ over a fixed evolution time, we apply an error recovery that only recovers photon number shift errors and leaves phase errors untouched. In this way, we ensure that we are not accidentally correcting a phase $\varphi$ that we would like to sense.

We choose a bosonic code known as a binomial code \cite{Marinoff2024} that has a discrete rotation symmetry. This is a natural choice because binomial codes have been shown to be effective against loss errors and because there is a natural way to decouple the photon number and phase recoveries. We make no claims about the performance of this encoding for estimating phase for metrological purposes; instead we want to see the emergence of an uncorrected phase error under nonlinear evolution. To see how this error correcting code works we plot the wigner function of the plus state through various steps of noisy metrology in \cref{fig:Correcting code}. There we see that the photon number recovery is sufficient to correct errors in the linear case, but is unable to recover the state completely in the nonlinear case. 

We will simultaneously apply the Hamiltonian evolution $U(\phi)$ and loss for different functions $g(\hat{n}) = \hat{n}^z$. We will relax the assumption that the probe state must not be in the kernel of the error, for simplicity. This will reduce the fidelity for all values of nonlinearity $z$. We will choose the loss rate $\kappa$ to be low enough that the linear case $z=1$ still has fidelity near one. We can then calculate the fidelity between the state that experienced the loss and recovery versus the state that never experienced the loss. We plot the results of this calculation in \cref{fig:correctedfidelity}. This reveals that the fidelity is near one for $z=1$ and quickly drops for $z>1$ even for relatively small probe states.

\begin{figure}[h]
  \centering
  \includegraphics[width=0.5 \textwidth]{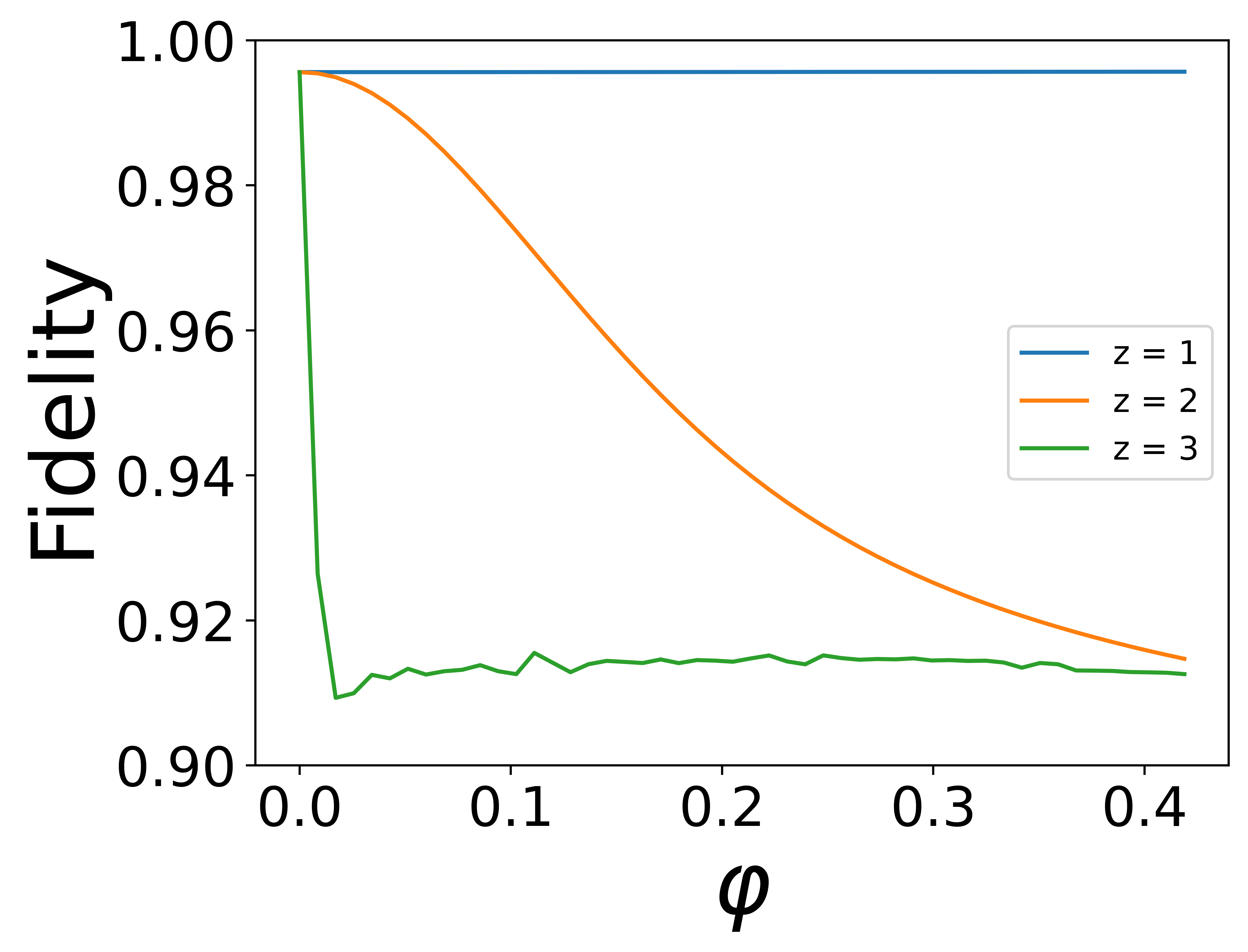}
  \caption{The fidelity of a probe that underwent noisy evolution followed by error correction with a probe that never encountered noise is plotted as a function of phase encoded. The same initial probe and loss channel were used as in \cref{fig:Correcting code}. For linear metrology, in blue, we see that the fidelity is $\approx .99$ regardless of phase. If the error correcting code could perfectly correct loss, then this fidelity would be one. For higher nonlinear powers we see increasingly sharp cutoffs with encoded phase $\varphi$. These fidelities asymptote to $\approx .9$ which closely matches with the $\approx 89\%$ chance that no photons are lost given our noise channel. This supports the assertion that beyond the critical phase the fidelities of an state with a shift error quickly drops to zero.  }
  \label{fig:correctedfidelity}
\end{figure}

\section{Errors on Specific States}

Lastly, we will explicitly evaluate the fidelity which we use as a figure of merit.
This is given by
\begin{equation}
    F(k,g,\psi) = |\langle\psi|\hat{V}_k(\varphi)|\psi\rangle|^2,
\end{equation}
which is Eq. (14) in the main text.
Here we show the results for a $0-N$ state, a coherent state, and a cat state.
The results are shown in~\cref{tab:summary}, where we also list the explicit form of the un-normalized state, and the QFI in the absence of errors.
We leave the normalization off in the table to improve legibility.

\begin{table*}[h]
\begin{tabular}{l|l|l|l}
State                  & QFI & $F(k,g,\psi)$ \\ \hline
$\ket{0}+\ket{N} \quad $      & $N^z$    & $ \displaystyle{ \frac{1}{2}\left(\cos[\theta(g(N+k)-g(N)-g(k))]+1\right) } $ \\ \hline
$\ket{\alpha}$         & $\ \displaystyle {4 z^2|\alpha|^{4z-2} + O(|\alpha|^{4z-4})} \quad  $    &  $ \displaystyle {\Big|\mathrm{e}^{-|\alpha|^2}\sum_n \frac{|\alpha|^{2n}}{n!}\mathrm{e}^{-i\theta( g(n+k)-g(n))}\Big|^2 }$                  \\ \hline
$\ket{0}+\ket{\alpha}$ &     $ 2 z^2|\alpha|^{4z-2} + O(|\alpha|^{4z-4}) $ &  $\displaystyle{ \frac{1}{4}\Big|\mathrm{e}^{-|\alpha|^2}\sum_n \frac{|\alpha|^{2n}}{n!}\mathrm{e}^{-i\theta( g(n+k)-g(n))} + \mathrm{e}^{i\theta g(k)}(2\mathrm{e}^{-|\alpha|^2/2}+1)\Big|^2 } \quad  $ 
\end{tabular}
\caption{A table of three common states, their QFI, and the fidelity of the emergent phase error}\label{tab:summary}
\end{table*}

\end{widetext}


\begin{thebibliography}{53}
\bibitem [{Note2({\natexlab{b}})}]{normNote}%
  \BibitemOpen
  \bibinfo {note} {If we don't assume $|\langle \Psi | \Psi \rangle|^2 = 1$,
  then $\langle \Psi | \hat{V}_k(\varphi) | \Psi \rangle = |\langle \Psi | \Psi
  \rangle|^2 - \mathcal{O}(\epsilon)$. The interpretation that resultant
  probability distributions will be at most $\epsilon$ different remains,
  except now the difference is between two quantum channels rather than between
  two states.}\BibitemShut {Stop}%
\end{thebibliography}
\end{document}